\documentclass[smallextended]{svjour2}                   
\smartqed  
\usepackage{graphicx}
\usepackage{amsmath}
\usepackage{mathtools}
\usepackage{amsfonts}
\usepackage{mathrsfs}
\newcommand{\intbar}{{- \hspace{-1 em}} \int}
\newtheorem{thm}{Theorem}

\begin{document}

\title{Non-Equilibrium Steady States in Kac's Model Coupled to a Thermostat
}


\author{Josephine Evans \thanks{The authors were supported by the UK Engineering and Physical
Sciences Research Council (EPSRC) grant EP/H023348/1 for the
University of Cambridge Centre for Doctoral Training, the Cambridge
Centre for Analysis.}
}


\institute{J. Evans \at
              fDepartment of Pure Mathematics and Mathematical Statistics\\
University of Cambridge\\
Wilberforce Road\\
Cambridge CB3 0WA, UK\\
              \email{jahe2@cam.ac.uk}         
}

\date{Received: date / Accepted: date}

\maketitle

\begin{abstract}
This paper studies the existence, uniqueness and convergence to non-equilibrium steady states in Kac's model with an external coupling. We work in both Fourier distances and Wasserstein distances. Our methods work in the case where the external coupling is not a Maxwellian equilibrium. This provides an example of a non-equilibrium steady state. We also study the behaviour as the number of particles goes to infinity and show quantitative estimates on the convergence rate of the first marginal.
\end{abstract}

\section{Introduction}
Kac's model was introduced by Mark Kac in 1956 \cite{KAC56}. It is a stochastic N-particle model designed to mimic the dynamics of velocities of particles in a spatially homogeneous dilute gas. The dynamics are those of N particles with one dimensional velocities, these particles interact in a Markov process, where two particles ``collide'' resulting in a mixing of their velocities. The state of the system can be described by the vector of velocities of each of the particles. Kac derived an equation on the law of this system, this equation is usually called the Kac master equation and it is a linear integro-differential equation. Kac showed that, in a certain sense, as the number of particles goes to infinity the master equation tends to a Boltzmann like equation. This motivates estimates on the behaviour of the marginals of solutions which are uniform in the number of particles, which could then be used to show, or at least indicate, the same behaviour for the Boltzmann equation. In general a direct study of the Boltzmann equation has proved more fruitful, however the master equation has become an object of study in its own right. Convergence to equilibrium and spectral gaps have been studied in Kac's master equation in both entropy \cite{CGLV10,EINAV11} and $L^2$ \cite{J01,CCL00}. This paper studies convergence to equilibrium for solutions of the master equation coupled to a thermostat. More precisely, we study the master equation for a system of $N$ particles who, as well as ``colliding'' with each other, can also ``collide'' with some infinite collection of other particles whose velocities lie in some fixed distribution. When this fixed distribution is not a Maxwellian this allows for the possibility of a non-equilibrium steady state. One possible more physical interpretation of this would be if the system was interacting with two different heat baths at different temperatures. Situations related to the existence and convergence to non-equilibrium steady states are studied in \cite{BLF2000,DHAR08,GLP78,RPE99,RT00} and in particular looking at exponential convergence in \cite{RT02,EH00}.  

This paper is fundamentally motivated by two others the first \cite{BLV14} studies a similar model but only in the situation where the thermal bath is a Maxwellian distribution. They show exponential convergence to equilibrium in both entropy and $L^2$. The second \cite{CLM15} studies the existence of non-equilibrium steady states in various coupled equations arising from mathematical physics including the non-linear spatially homogeneous Boltzmann equation. The paper \cite{BLV14} suggest as a further question, what would happen in the case of a non-Maxwellian reservoir and we adapt the techniques of \cite{CLM15} to study this situation. We also include a study of how our estimates on the first marginal behave as the number of particles $N \rightarrow \infty$. This allows us, in some sense, to commute the long time and $N \rightarrow \infty$ limit. The $N \rightarrow \infty$ limit is very similar to the equations studied in \cite{CLM15}, they study a coupled Boltzmann equation where in our case the limit would be a coupled Boltzmann-Kac equation. The convergence, both in this paper and in the Maxwellian case studied in \cite{BLV14}, is primarily driven by the external force and not by the Kac mixing part. However, the effect of the Kac part is more evident in this paper since it affects the form of the steady state. The work in \cite{BLV14} has been extended in \cite{TV15, BLTV16} to study how their thermostatted model relates to a partially thermostatted model and to the original Kac's model. In this second paper they make use of the GTW distance used in our work.

  Following the strategy of \cite{CLM15} we study the problem of convergence to equilibrium in the Gabetta-Toscani-Wennberg metric . This metric is introduced in \cite{GTW85} and is
\[ d_{GTW, N}(f,h) = \sup_{\xi \in \mathbb{R}^N, \xi \neq 0}\frac{|\hat{f}(\xi)-\hat{h}(\xi)|}{|\xi|^2}, \] where $\hat{f}$ represents the Fourier transform of $f$. This is a metric on the space of probability measures with finite second moment and the same finite first moment. We also study convergence in the metric 
\[ d_{T1,N}(f,h)=\sup_{\xi \in \mathbb{R}^N, \xi \neq 0} \frac{|\hat{f}(\xi)-\hat{h}(\xi)|}{|\xi|}, \] This is a metric on the space of probability distributions with finite mean.

 If we choose $g$ to be the distribution of the particles in the thermostat and we pick $g \in L^2$ such that $g$ is a probability distribution function with zero mean and finite second moment $K_g$ then the master equation for the system we study is
\begin{align} \partial_t F_n = -\lambda N(I-Q)[F_N] - \mu \sum_{j=1}^N (I-R_j)[F_N] = \mathcal{L}[F_N], \label{mastereq} \end{align} where
\[ Q[F_N]= \frac{1}{{N \choose 2}} \sum_{i<j} \intbar_0^{2\pi}F_N(v_{i,j}(\theta)) \mathrm{d}\theta, \] and \[ R_j[F_N]=\int \mathrm{d}w \intbar_0^{2\pi} \mathrm{d}\theta g(w_j^*) F_N(v_j(w, \theta)). \] In these
\begin{align*}
v_{ij}(\theta)&=(v_1,\dots,v_i \cos (\theta)+v_j \sin (\theta),\dots,-v_i \sin (\theta)+v_j \cos (\theta),\dots,v_N),\\
v_j(w,\theta)&=(v_1,\dots,v_j\cos (\theta) + w \sin (\theta),\dots,v_N),\\
w_j^*&=w\cos (\theta)-v_j \sin (\theta).
\end{align*}

We show that
\begin{thm}
A steady state for the master equation exists, is unique and has the same moments up to order 2 as $g^{\otimes N}$. \label{fixedpoint}
\end{thm}
\begin{thm} \label{convergence}
If we start with initial data $F^0_N$ and $H^0_N$ which are probability distributions on $\mathbb{R}^N$ with finite first and second moments then we have the following possible situations:

1. If $F^0$ and $H^0$ have the same mean initially then the $GTW$ distance between the solutions is finite for all time and we get the exponential convergence:
\[ d_{GTW,N} (F_N(t), H_N(t)) \leq e^{-\mu t/2} d_{GTW, N}(F_N^0, H_N^0). \]

2. If $F^0$ and $H^0$ have different means then we can construct an altered distance in which the solutions still converge exponentially fast towards each other with rate $\mu/2$. We also have the estimate
\[ d_{T1, N}(F_N(t), H_N(t)) \leq e^{-\mu t/4} d_{T1,N}(F_N^0, H_N^0). \]
\end{thm} 
\begin{remark}The altered distance involves adding a correction term and is defined in order to deal with the fact that the $GTW$ distance cannot deal with initial data with non-zero mean. If the two solutions initially have the same mean this reduces to the $GTW$ distance. We give the theorem in both distances which shows we can either sacrifice something in the dependence on initial data or in the rate. In the asymptotic study as $N \rightarrow \infty$ the two distances give the same dependence on $N$ through different mechanisms which suggests that the dependence on $N$ occurring here is in some way intrinsic to the problem.
\end{remark}
\begin{remark}
Here $\mu/2$ is the rate found in \cite{BLV14} to be the $L^2$ spectral gap and the rate of convergence to equilibrium in relative entropy.
\end{remark}
Furthermore we wish to study how the $N$ particle Kac's model behaves as $N \rightarrow \infty$ in the manner originally proposed by Kac to link it with the spatially homogeneous Boltzmann equation. In order to do this we study how the convergence results which we have obtained can be translated into convergence results on the first marginal. We prove properties of the GTW metric which are similar to subadditivity. If the initial data $(F_N(0))_{N \geq 2}$ forms a chaotic family then we can control the convergence rate of the first marginals to equilibrium uniformly in $N$. We formally define the notion of chaotic family later. Similarly to \cite{BLV14} we can prove propagation of chaos in exactly the same manner as Kac in \cite{KAC56}. This means that the first marginals of the solution to the master equation will limit to the solution of a Boltzmann like equation. This motivates our proof of uniform in $N$ convergence rates for the first marginal.
\begin{thm} Suppose that $f$ and $h$ are mean zero probability densities on $\mathbb{R}$. If $(F_N(0,v))_{N \geq 2}$ and $(H_N(0,v))_{N \geq 2}$ are respectively $f,h$-chaotic families with respect to the Gabetta-Toscani-Wennberg metric. If furthermore, the distance between $F_N(0,\cdot)$ and $f^{\otimes N}$, and between $H_N(0,\cdot)$ and $h^{\otimes N}$ are bounded uniformly in $N$, and $F_N, H_N$ are the solution to the $N$-particle coupled Kac's master equation with this initial data then there exists a constant $C$ independent of $N$ such that
\[ d_{GTW,1} ( \Pi_1(F_N),\Pi_1(H_N)) \leq (C+d_{GTW,1}(f,h)) e^{-\frac{\mu}{2}t}. \] \label{chaos}
\end{thm}
Here we say that a family is $f$-chaotic with respect to a family of metrics, $(d_k)$, if
\[ d_{k}(\Pi_k[F_N],f^{\otimes N}) \rightarrow 0, \] as $N \rightarrow 0$ for every $k$. Here $d_k$ is a metric on $\mathbb{R}^k$ and $\Pi_k$ is a projection onto this subspace of $\mathbb{R}^N$. This is the standard notion of chaoticity which was introduced by Kac. Here we write it in terms of a distance which metrizes weak convergence of measures as it is more convenient for our set up.
\begin{remark}
Our theorem is really designed to work in the case of tensorised initial data and can be extended slightly as we have shown. If we no longer wanted our estimates to depend on the first marginal of the initial data we could replace it with the weaker, but difficult to check, condition
\[  d_N(F_N, H_N) \leq C \hspace{10pt} \forall N.\]
\end{remark}

We also have two theorems in the case where we have non-zero and non equal mean for $f$ and $h$ using each of the different metrics which we use to study this case.

\begin{thm} \label{convwithcorrectionterm}
Let $F_N^0$ and $H_N^0$ are respectively $f$ and $h$ chaotic families where the $GTW$ distance between $F_N^0$ and $f^{\otimes N}$ (resp. for $H_N^0$ and $h^{\otimes N}$) is bounded uniformly in $N$. Furthermore if $f$ and $h$ are probability densities with finite first and second moments and differentiable Fourier transforms, then we can choose a family of functions $\chi$ (one for each $N$) to construct an altered distance $\tilde{d}$ so that
\[ \tilde{d} \left( \Pi_1[F_N], \Pi_1[H_N]) \right) \leq (C_1 + (C_2 + C_3) \sqrt{N} + \tilde{d}(f,h))e^{-\frac{\mu}{2}t}. \] 
\end{thm}

\begin{thm} \label{t1distfirstmarginal}
Suppose that $f$ and $h$ are probability densities on $\mathbb{R}$ with finite mean. If $(F_N(0,v))_{N \geq 2}$ and $(H_N(0,v))_{N \geq 2}$ are respectively $f,h$-chaotic families with respect to the $T1$ metric, and the $T1$ distance 
between $F_N(0,\cdot)$ and $f^{\otimes N}$, and between $H_N(0,\cdot)$ and $h^{\otimes N}$ are bounded uniformly in $N$. Furthermore, let $F_N, H_N$ are the solution to the $N$-particle coupled Kac's master equation with this initial data, then there exists a $C$ (the bound between the initial data and the tensorised form) of $N$ such that 
\[ d_{T1,1}(\Pi_1[F_N](t), \Pi_1[H_N](t)) \leq (C + \sqrt{N} d_{T1, 1}(f,h))e^{-\mu t/4}. \]
\end{thm}

We can also prove two similar theorems in Wasserstein distance on measures with finite second moment. The Wasserstein distance is given by
\[ \mathcal{W}_{2,d}(\mu, \nu)= \inf_{\pi}\left(\int_{\mathbb{R}^{2d}} \|\mathbf{x}-\mathbf{y}\|^2 \pi(\mathrm{d}\mathbf{x}, \mathrm{d}\mathbf{y}) \right)^{1/2}, \] here $\pi$ ranges over measures with marginals $\mu, \nu$.
\begin{thm}
If $\mu_N$ and $\nu_N$ are two solutions to the master equation with finite second moments then
\[ \mathcal{W}_2(\mu_N(t), \nu_N(t)) \leq e^{-\mu t/2} \mathcal{W}_2(\mu_N(0), \nu_N(0)). \] \label{wassersteincontraction}
\end{thm}
\begin{thm}
Suppose that $\mu_N(t)$ and $\nu_N(t)$ are solutions to the master equation at time $t$, with initial data $\mu_0^{\otimes N}$ and $\nu_0^{\otimes N}$ then we have that for any $N$,
\[ \mathcal{W}_{2,1}(\Pi_1(\mu_N(t)), \Pi_1(\nu_N(t))) \leq e^{-\mu t/2} \mathcal{W}_{2,1}(\mu_0, \nu_0).  \] \label{wassersteinmarginals}
\end{thm}

\section{Behaviour of the Moments}
In this section we prove some basic lemmas on how the moments of a solution behave. We recall that $K_g$ is the second moment of $g$ our fixed distribution.
\begin{lemma}
The kinetic energy of a solution to the coupled master equation converges exponentially fast to $NK_g$ with rate $\mu/2$. \label{energylemma}
\end{lemma}
\begin{proof}
Let
\[ K(t) = \int_{\mathbb{R}^n} \|v\|^2 F_N(v) \mathrm{d}v. \] Differentiating under the integral and recalling that radial functions are in the kernel of $(I-Q)$ and that $(I-Q)$ is self adjoint we get,
\[ \partial_t K = \mu \sum_{j=1}^N \int_{\mathbb{R}^N} \mathrm{d}v \int \mathrm{d}w \intbar_0^{2\pi}\mathrm{d}\theta g(w_j^*)F_N(v_j(w,\theta))\|v\|^2 - \mu N K. \] The Jacobian of the change of variables $(v_j(w,\theta),w_j^*) \leftrightarrow (v,w)$ is $1$. Also we have that $\|v\|^2 + w^2 = \|v_j(w,\theta)\|^2 +w_j^{*2}$. Using these we have

\begin{align*}
\partial_t K =& \mu \sum_{j=1}^N \int_{\mathbb{R}^N} \mathrm{d}v \int \mathrm{d}w \intbar_0^{2\pi} \mathrm{d}\theta g(w)F_N(v)(\|v\|^2 +w^2)\\
&-\mu \sum_{j=1}^N \int_{\mathbb{R}^N} \mathrm{d}v \int \mathrm{d}w \intbar_0^{2\pi} \mathrm{d}\theta g(w) F_N(v) w_j^{*2} - \mu N K,\\
&=\mu N K +\mu N K_g - \mu N K\\
&-\mu \sum_{j=1}^N \int_{\mathbb{R}^N} \mathrm{d}v \int \mathrm{d}w \intbar_0^{2\pi} \mathrm{d}\theta g(w) F_N(v)(w^2 \cos^2 \theta -2wv_j \cos \theta \sin \theta \\ & + v_j^2 \sin^2 \theta), \\
&=\mu N K_g -\mu N \frac{1}{2}K_g - \frac{\mu}{2}K,\\
&= - \frac{\mu}{2}(K-NK_g).
\end{align*}
\end{proof}
\begin{lemma}
The first moments of a solution to the coupled master equation converge to $0$ with rate greater than $\mu/2$. Also the second order moments \[ d_{k,l}=\int_{\mathbb{R}^N}F_N(v)v_k v_l \mathrm{d}v, \] converge to $0$ with rate greater than $\mu/2$. \label{firstmoments}
\end{lemma}
\begin{proof}
Let $d_k = \int \mathrm{d}v F_N(v) v_k$ then we get the equation
\begin{align*}
\partial_t d_k &= - N(\lambda + \mu) d_k + \lambda(N-2)d_k + \mu(N-1)d_k,\\
&=-(2\lambda + \mu)d_k.
\end{align*} For the second set we can calculate
\[ \partial_t d_{k,l} =\left (-4\lambda -2\mu + \frac{2\lambda}{N-1}\right)d_{k,l} \]
\end{proof}

\section{Existence, Uniqueness and Convergence to a Steady State}
We wish to show existence and uniqueness of a steady state via the Banach fixed point theorem in the space of probability measures with zero mean and finite second moment with the GTW distance. In order to do this we write the steady state equation for $F_N$ as a fixed point theorem. We set $\gamma = \lambda/(\lambda + \mu)$ to mirror the notation in \cite{CLM15}.
\[ F_N = \gamma Q[F_N] + (1-\gamma) \frac{1}{N} \sum_{j=1}^N R_j[F_N]=\Phi[F_N].\] We want to show that $\Phi$ is a contraction in the Gabetta-Toscani-Wennberg metric. We first need to show that $\Phi$ preserves the metric space that we are working in. 
\begin{lemma}
Suppose $F_N$ has mean zero and finite second moment then $\Phi[F_N]$ has mean zero and finite second moment.
\end{lemma}
\begin{proof}
\begin{align*} \int_{\mathbb{R}^N}Q[F_N]v_k \mathrm{d}v =& \frac{N-2}{N}\int_{\mathbb{R}^N}F_N(v)v_k \mathrm{d}v\\ & + \frac{1}{{N\choose 2}} \sum_{i<k} \int_{\mathbb{R}^N}\intbar_0^{2\pi}F_N(v)(v_i\cos \theta +v_k \sin \theta)\mathrm{d}\theta \mathrm{d}v \\
& + \frac{1}{{N \choose 2}} \sum_{k<j} \int_{\mathbb{R}^N}\intbar_0^{2\pi} F_N(v)(-v_k \sin \theta + v_j \cos \theta) \mathrm{d}\theta \mathrm{d}v,\\
=& \frac{N-2}{N} \int_{\mathbb{R}^N} F_N(v) v_k \mathrm{d}v = 0. \end{align*} It is immediate that $\int R_j[F_N](v)v_k \mathrm{d}v =0$ for $j \neq k$. So it remains to look at
\begin{align*} \int_{\mathbb{R}^N}\mathrm{d}v R_k[F_N](v)v_k &= \int_{\mathbb{R}^N} \int \mathrm{d}w \intbar_0^{2\pi}\mathrm{d}\theta g(w_j^*)F_N(v_j(w,\theta) v_k\\
&= \intbar_0^{2\pi}\mathrm{d}\theta \int_{\mathbb{R}^N} \int \mathrm{d}v \mathrm{d}w g(w)F_N(v)(v_k \cos \theta - w \sin \theta)=0.  \end{align*} The fact that $\Phi[F_N]$ has finite second moments is clear since $Q^*, R_j^*$ acting on $\|v\|^2$ or similar produces a finite linear combination of other functions to make second moments.
\end{proof}
Further we would like to calculate how $Q$ and $R_j$ act in Fourier space.
\begin{lemma} \label{operatorsinfourier}
\[ \widehat{Q[F_N]}(\xi) = \frac{1}{{N \choose 2}} \sum_{k<j} \intbar_0^{2\pi} \widehat{F_N}(\xi_{k,j}) \mathrm{d}\theta, \] where $\xi_{k,j} = (\xi_1, \dots,\xi_k \cos \theta + \xi_j \sin \theta, \dots, -\xi_k \sin \theta + \xi_j \cos \theta , \dots, \xi_N)$.
Also,
\[ \widehat{R_j[F_N]}(\xi) = \intbar_0^{2\pi} \widehat{F_N}(\xi_j(\theta)) \hat{g}(\xi_j \sin \theta) \mathrm{d}\theta, \] where $\xi_j(\theta)=(\xi_1 , \dots, \xi_j \cos \theta , \dots , \xi_N)$.
\end{lemma}
\begin{proof}
\begin{align*}
\int_{\mathbb{R}^N}Q[F_N]e^{-iv\cdot \xi}\mathrm{d}v &= \frac{1}{{N \choose 2}} \sum_{k<j} \intbar_0^{2\pi} \mathrm{d}\theta \int_{\mathbb{R}^N} \mathrm{d}v F_M(v_{kj}(\theta))e^{-iv\cdot \xi},\\ &= \frac{1}{{N \choose 2}} \sum_{k<j} \intbar_0^{2\pi} \mathrm{d}\theta \int_{\mathbb{R}^N}\mathrm{d}v F_N(v)e^{-iv_{k,j}(\theta)\cdot \xi},\\
&=(2\pi)^{N/2}\frac{1}{{N \choose 2}} \sum_{k<j} \intbar_0^{2\pi} \mathrm{d}\theta \widehat{F_N}(\xi_{k,j}).
\end{align*} Where $\xi_{k,j} = (\xi_1,\dots,\xi_k \cos \theta + \xi_j \sin \theta,\dots,-\xi_k \sin \theta + \xi_j \cos \theta,\dots, \xi_N)$.
\begin{align*}
\int_{\mathbb{R}^N}\mathrm{d}v R_j[F_N] e^{-iv\cdot \xi} &= \intbar_0^{2\pi} \mathrm{d}\theta \int \mathrm{d}w \int_{\mathbb{R}^N}\mathrm{d}v g(w_j^*) F_N(v_j(w,\theta))e^{-iv\cdot \xi}\\
&=\intbar_0^{2\pi} \mathrm{d}\theta \int \mathrm{d}w \int_{\mathbb{R}^N} \mathrm{d}v g(w) F_N(v) e^{-iv_j(w,\theta)\cdot \xi}\\
&=(2\pi)^{N/2} \intbar_0^{2\pi}\mathrm{d}\theta \widehat{F_N}(\xi_j(\theta))\hat{g}(\xi_j \sin \theta).
\end{align*} Where $\xi_j(\theta) = (\xi_1,\dots,\xi_j \cos \theta,\dots,\xi_N)$. 
\end{proof}
Now we can show existence and uniqueness.
\begin{proof}[Proof of Theorem \ref{fixedpoint}]
Calculating we have
\begin{align*} &\widehat{\Phi[F_N]}(\xi)=\\&\frac{1}{(2\pi)^{N/2}}\left( \gamma \int_{\mathbb{R}^N} Q[F_N](v)e^{-v\cdot \xi}\mathrm{d}v + (1-\gamma)\frac{1}{N}\sum_{j=1}^N \int_{\mathbb{R}^N} R_j[F_N]e^{-iv\cdot \xi} \mathrm{d}v \right). \end{align*}
Using the results of \ref{operatorsinfourier} we have
\[ \widehat{\Phi[F_N]} = \intbar_0^{2\pi} \mathrm{d}\theta \left( \frac{\gamma}{{N \choose 2}} \sum_{i<j} \widehat{F_N}(\xi_{i,j}(\theta)) + \frac{1-\gamma}{N} \sum_{j=1}^N \widehat{F_N}(\xi_j(\theta)) \hat{g}(\xi_j \sin \theta) \right). \] Therefore
\begin{align*}
&\sup_{\xi \neq 0} \frac{|\widehat{\Phi[F_N]}(\xi)-\widehat{\Phi[H_N]}(\xi)|}{|\xi|^2}\\
& \leq \sup_{\xi \neq 0} \frac{|\widehat{F_N}(\xi)-\widehat{H_N}(\xi)|}{|\xi|^2} \intbar_0^{2\pi} \mathrm{d}\theta \left(\frac{\gamma}{{N \choose 2}} \sum_{i<j} \frac{|\xi_{i,j}(\theta)|^2}{|\xi|^2}+ \frac{1-\gamma}{N} \sum_{j=1}^N \hat{g}(\xi_j \sin \theta) \frac{|\xi_j(\theta)|^2}{|\xi|^2}  \right)\\
& \leq \left( \gamma + \frac{1-\gamma}{N}\left(N-\frac{1}{2}\right)  \right) d_{GTW}(F_N, H_N)\\
& \leq \left(1- \frac{1-\gamma}{2N}\right) d_{GTW}(F_N, H_N).
\end{align*} 
Here to go between the second and third line we used
\begin{align*}
\sum_{j=1}^N \hat{g}(\xi_j \sin \theta) \frac{|\xi_j(\theta)|^2}{|\xi|^2} & \leq \sum_{j=1}^N \frac{|\xi_j(\theta)|^2}{|\xi|^2} \\ & = \sum_{j=1}^N \frac{|\xi|^2 - \xi_j^2 \sin^2 \theta}{\|xi|^2} \\&= N-\sin^2 \theta.
\end{align*}
So we have the required contraction property for any fixed $N$. Which shows existence and uniqueness of a steady state thanks to the contraction mapping theorem. The moments being the same up to order 2 as $g$ follow from the lemmas on the behaviour of moments in the previous section.
\end{proof}
We also want to prove a contraction estimate in the $T1$ distance.
\begin{lemma} \label{t1contraction}
\[ d_{T1,N}(\Phi[F_N], \Phi[H_N]) \leq \left( 1-\frac{1-\gamma}{4N} \right)d_{T1, N}(F_N, H_N). \]
\end{lemma}
\begin{proof}
The proof is the same as for the $GTW$ distance but here it is necessary to use 
\[ (1-x^2)^{1/2} \leq 1-\frac{1}{2}x^2, \] when bounding $|\xi_j(\theta)|/|\xi|$. This time we have
\begin{align*}
\sum_{j=1}^N \hat{g}(\xi_j \sin \theta)\frac{|\xi_j(\theta)|}{|\xi|} & \leq \sum_{j=1}^N \sqrt{\frac{|\xi|^2- \xi_j^2 \sin^2 \theta}{|\xi|^2}}\\ & \leq \sum_{j=1}^n \left( 1 - \frac{1}{2} \frac{\xi_j^2 \sin^2 \theta}{|\xi|^2} \right) \\& = N - \frac{1}{2}\sin^2 \theta.
\end{align*}
\end{proof}
Using these estimates we can also show convergence to equilibrium.
\begin{proof}[Proof of Theorem \ref{convergence}] Suppose initially that $F_N(t)$ and $H_N(t)$ both have zero mean. From the above calculation we have
\begin{align*} F_N(t+s)-H_N(t+s)=& (1-s(\lambda + \mu)N)(F_N(t)-H_N(t)) \\&+ s(\lambda + \mu) N ( \Phi[F_N(t)]-\Phi[H_N(t)]) + o(s). \end{align*} Therefore
\begin{align*}
d_{GTW}(F_N(t+s)&,H_N(t+s)) \leq (1-s(\lambda + \mu)N) d_{GTW}(F_N(t), H_N(t))\\& + s(\lambda + \mu)N d_{GTW}(\Phi[F_N], \Phi[H_N]) + o(s) \\
\leq & (1-s(\lambda + \mu))d_{GTW}(F_N(t),H_N(t)) \\&+ s(\lambda + \mu)N\left( 1- \frac{1-\gamma}{2N}  \right) d_{GTW}(F_N(t), H_N(t)) + o(s)\\
=&\left(1-\frac{\mu}{2}s\right) d_{GTW}(F_N(t), H_N(t)) + o(s).
\end{align*}
Hence,
\[  \frac{\mathrm{d}}{\mathrm{d}t} d_{GTW}(F_N(t),H_N(t)) \leq - \frac{\mu}{2} d_{GTW}(F_N(t), H_N(t)).\] So that we have exponential decrease with the stated rate. Since in \ref{firstmoments} we showed that if we start the dynamics with two distribution which have zero mean then this property will be preserved, we see that if we start the dynamics with a zero mean distribution then it will converge exponentially fast towards the steady state. Now we would like to add a correction term so that we can deal with a wider class of initial data as in \cite{CLM15}. We define
\[ \widehat{\mathcal{M}[F_N]} := \chi (\xi) \sum_{k=1}^N \left( \int_{\mathbb{R}^N}v_k F_N(v)\mathrm{d}v   \right)i \xi_k, \] where $\chi$ is a smooth, compactly supported function which is 1 in some neighbourhood of 0. Therefore, if $D_N = F_N - H_N - \mathcal{M}[F_N-H_N]$ we will have that
\[ \widehat{D_N} = \int_{\mathbb{R}^N} \mathrm{d}v \left(F_N(v)-H_N(v) \right)\left(e^{-iv \cdot \xi} - \chi (\xi) \sum_{j=1}^N v_j \xi_j  \right). \] This means that
\[ \sup_{\xi \neq 0} \frac{\widehat{D_N}(\xi)}{|\xi|^2}< \infty. \] We calculate that
\begin{align*}
\partial_t D_N =& \partial_t F_N - \partial_t H_N - \partial_t \mathcal{M}[F_N-H_N] \\ =& \lambda N(I-Q)[D_N] - \mu \sum_{j=1}^N (I-R_j)[D_N]\\
&- \lambda (I-Q)[\mathcal{M}[F_N - H_N]] - \mu \sum_{j=1}^N(I-R_j)[\mathcal{M}[F_N-H_N]] \\&- \partial_t \mathcal{M}[F_N - H_N].
\end{align*} So if we let
\[ W = -\lambda N(I-Q)[\mathcal{M}[F_N-H_N]] - \mu \sum_{j=1}^N (I-R_j)[\mathcal{M}[F_N-H_N]] - \partial_t \mathcal{M}[F_N - H_N], \] then $D_N$ is a zero momentum, zero integral function and we have the equation
\[ \partial_t D_N = -(\lambda + \mu)N(D_N-\Phi[D_N]) + W. \] So if we want to show that
\[ \sup_{\xi \neq 0} \frac{|\widehat{D_N}|}{|\xi|^2}, \] converges to zero exponentially fast it is sufficient to show that,
\[ \sup_{\xi \neq 0} \frac{|\widehat{W}(\xi)|}{|\xi|^2}, \] converges to zero exponentially fast. Since $\partial_t$ commutes with Fourier transform and $\chi$ is compactly supported we know that
\[ \widehat{\mathcal{M}}[F_N-H_N] = \chi(\xi) \sum_{k=1}^N (m_f(0)-m_h(0))e^{-(2\lambda + \mu)t}i\xi_k, \] So ignoring $\chi$ and looking near 0 we have, after Taylor expanding and using the formula from lemma \ref{operatorsinfourier}
\begin{align*}
& -\lambda N \widehat{(I-Q)[\mathcal{M}]} - \mu \sum_{j=1}^N \widehat{(I-R_j)[\mathcal{M}]} = \\
& -(2 \lambda + \mu)(m_f(0)-m_h(0)) e^{-(2\lambda + \mu)t} \sum_{k=1}^N \xi_k \\& - \frac{1}{2}\mu K_g (m_f(0)-m_h(0))e^{-(2\lambda + \mu)t} |\xi|^2 \sum_{k=1}^N \xi_k + o(|\xi|^3).
\end{align*} Therefore near $\xi = 0$, we have
\[ \frac{\widehat{W}(\xi)}{|\xi|^2} = -\frac{1}{2} \mu K_g \sum_{k=1}^N \xi_k + \frac{1}{2}\mu K_g \frac{\sum_{k=1}^N \xi_k^3}{|\xi|^2} + o(\xi). \] This is because the lower order terms cancel.
So in particular we have that 
\[ \lim_{\xi \rightarrow 0} \frac{\widehat{W}(\xi)}{|\xi|^2} = 0. \] Therefore, since $\widehat{W}$ has compact support we can bound 
\[ \frac{\widehat{W}(\xi)}{|\xi|^2} \leq C e^{-(2 \lambda + \mu)t} \] where $C$ may increase with $N$. At 0 the gradient of 
\[ w(\xi)=\frac{\hat{W}(\xi)}{|\xi|^2} \] is $C\sqrt{N}\mu K_g /2$ so the gradient of $w$ cannot be bounded uniformly in $N$. Since we can calculate $w(\xi)$ explicitly if $\chi$ is always radial as
\[ \mu \left( 1 - \sum_{j=1}^N(1-\alpha_j(\xi)) \right) \frac{\mathcal{M}}{|\xi|^2} \] where
\[ \alpha_j(\xi)= \intbar_0^{2\pi} \left( 1-\frac{\xi_j (1-\cos \theta)}{\sum_k \xi_k}\right) \hat{g}(\xi_j \sin \theta ) \frac{\chi(\xi_j(\theta))}{\chi (\xi)} \mathrm{d}\theta. \] This can be bounded uniformly provided we can bound the ration of the $\chi$s. Therefore under these additional assumptions we see that $w$ increases no faster than $\sqrt{N}$. This will give that
\[ \sup_{\xi \neq 0} \frac{|\widehat{D_N(t)}|}{|\xi|^2} \leq \left( C \sqrt{N} + \frac {|\widehat{D_N}(0)|}{|\xi|^2}\right) e^{-\frac{\mu}{2}t}. \] Therefore if we define a new distance
\[ \tilde{d}_N(F_N, H_N) = \sup_{\xi \neq 0} \frac{|\widehat{D_N}|}{|\xi|^2} + \sup_{\xi \neq 0} \frac{|\widehat{W}|}{|\xi|^2}, \] we will get the inequality
\[ \tilde{d}_N(F_N(t),H_N(t)) \leq Ce^{-\frac{\mu}{2}t}. \]
For the exponential convergence in the $T1$ distance we use the same argument as for the $GTW$ distance with the same mean and the contraction estimate in Lemma \ref{t1contraction}.
\end{proof}
\begin{remark}
If it were possible to get a bound on $|\nabla w(\xi)|$ in terms of $\sqrt{N}$ then it might in fact allow us to choose $\chi$ for each $N$ such that we didn't get the increase with $N$ by letting the radius of the support of $\chi$ decrease with $\sqrt{N}$. However, since the goal is to control the behaviour as $N \rightarrow \infty$ then in the case of different marginals working with the correction term would introduce an error of at least $\sqrt{N}$ when trying to control the initial data by its first marginal. In general because of having to choose a $\chi$ for each $N$ the altered distance is not well adapted to asymptotic analysis. We include it to show that for each $N$ we can get the rate $\mu/2$ and to compare with the limit equation case which is studied using this method in \cite{CLM15}.
\end{remark}
\section{Convergence Rate of the First Marginal} It is shown in \cite{BLV14} that propagation of chaos holds for this type of coupled Kac's model. The argument is very similar to Kac's original argument therefore is not repeated here. Since we have propagation of chaos we know that the first marginal of $F_N(t)$ will converge weakly towards a solution of the Boltzmann-Kac equation. In some sense we would like to be able to understand the two limits $t \rightarrow \infty$ and $N \rightarrow \infty$ simultaneously. For this reason we prove a bound on convergence to equilibrium for the first marginal which is uniform in $N$. Unfortunately, the $GTW$ distance and our correction term $W$ behave differently as $N \rightarrow \infty$ so it was only possible to get these estimates when the initial data has zero mean. 

The functions we work with will be invariant under permutations of variables so we can define the $k^{th}$ marginal for $k \leq N$
\[ \Pi_k[F_N] := \int_{\mathbb{R}^{N-k}} F_N(v_1,\dots,v_N)\mathrm{d}v_{i_1}\dots\mathrm{d}v_{i_{N-k}} \] for any choice of $1 \leq i_1 < i_2 < \dots < i_{N-k} \leq N$. Many of the distances in which we could study Kac's model, typically weighted $L^2$ distances will not behave well as the number of particles tends to infinity so will not give convergence of the first marginal to an equilibrium in entropy, here the subadditivity property of entropy in the number of variables is crucial. We wish to show that the GTW and related distances will possess similar subadditivity properties, which will allow us to control things in a similar way.
\begin{lemma}
\[d_{GTW,k}(\Pi_k[F_N], \Pi_k[H_N]) \leq d_{GTW,N}(F_N, H_N),\]
\[ \tilde{d}_k (\Pi_k[F_N], \Pi_k[H_N]) \leq \tilde{d}_k(F_N, H_N), \]
and
\[ d_{T1,k}(\Pi_k[F_N], \Pi_k[H_N]) \leq d_{T1,N}(F_N, H_N). \] \label{controlmarginals}
\end{lemma}
\begin{proof} The proof is the same for all the distances so we only do it in the case of $GTW$.
We can notice that
\[ \widehat{\Pi_k[F_N]}(\xi_1,\dots,\xi_k) = \widehat{F_N}(\xi_1,\dots,\xi_k,0,\dots,0). \] Using this we have that
\begin{align*}
d_{GTW,k}(\Pi_k[F_N], \Pi_k[H_N]) =& \sup_{\xi \neq 0, \xi_{k+1}=\dots=\xi_N=0} \frac{|\widehat{F_N}(\xi)-\widehat{H_N}(\xi)|}{|\xi|^2}\\
\leq & \tilde{d}(F_N,H_N).
\end{align*}
\end{proof}
\begin{lemma} If $f, h$ have the same first moments
\[ d_{GTW,N}(f^{\otimes N}, h^{\otimes N}) = d_{GTW,1}(f,h) \] where $d_{GTW,k}$ is the GTW distance on probability densities with $k$-variables. \label{tensorisedfunctions}
\end{lemma}
\begin{proof}
\begin{align*}
d_{GTW}(f^{\otimes N}, &h^{\otimes N}) = \sup_{\xi \neq 0} \frac{|\hat{f}(\xi_1)\dots\hat{f}(\xi_N)-\hat{h}(\xi_1)\dots\hat{h}(\xi_N)|}{|\xi|^2}\\
\leq & \sup_{\xi \neq 0} \frac{\sum_{i=1}^N|\hat{f}(\xi_1)\dots\hat{f}(\xi_{i-1})(\hat{f}(\xi_i)-\hat{h}(\xi_i))\hat{h}(\xi_{i+1})\dots\hat{h}(\xi_N)|}{|\xi|^2}\\
\leq & \sup_{\xi \neq 0} \sum_{i=1}^N \frac{\hat{f}(\xi_i)-\hat{h}(\xi_i)}{\xi_i^2} \frac{\xi_i^2}{|\xi|^2} \\
\leq & \sup_{\xi \neq 0} \sum_{i=1}^N d_{GTW,1}(f,h)\frac{\xi_i^2}{|\xi|^2} = d_{GTW,1}(f,h).
\end{align*} Since $f,h$ are the first marginals of $f^{\otimes N}, h^{\otimes N}$ respectively we have by the earlier lemma that
\[ d_{GTW,1}(f,h) \leq d_{GTW,N}(f^{\otimes N}, h^{\otimes N}) \] putting the two inequalities together gives the required result.
\end{proof}
We have already seen that 
\[ \frac{\widehat{W}(\xi)}{|\xi|^2}, \] may increase with $N$ so this will cause us problems if we wished to try and control $\tilde{d}_N(f^{\otimes N}, h^{\otimes N})$ by $\tilde{d}_1(f,h)$ . Even given this it would be good to be able to push the control by first marginals to general functions. However, the next lemma shows that this is not possible.
\begin{lemma}
There exist $f,g$ with finite second moment such that $f,g$ are symmetric and mean zero and they have the same marginals but $f,g$ are not the same. This means we cannot control the GTW distance between $f$ and $g$ in terms of the GTW distance between their first marginals. \label{restriction}
\end{lemma}
\begin{proof}
Let $\phi$ be a density function on $\mathbb{R}$ which is mean zero but not even. Define
\[ f(v_1, v_2) := \frac{1}{2} (\phi(v_1)\phi(-v_2)+\phi(-v_1)\phi(v_2)), \] and \[ g(v_1,v_2)= \frac{1}{2}(\phi(v_1)\phi(v_2) + \phi(-v_1)\phi(-v_2)). \] Then it is easy to see that $f$ and $g$ have the required properties.
\end{proof}
We wish to combine these lemmas in such a way as to get uniform control on the first marginal. Given the restriction shown by Lemma \ref{restriction} we want to choose `good' initial data in order that the distance between the initial data is controlled by the distance between the first marginals.
\begin{proof}[Proof of Theorem \ref{chaos}]
Since $f,h$ have mean zero and the GTW distance between $F_N(0)$ and $f^{\otimes N}$ is finite, we have that $F_N$ and $H_N$ have zero mean initially. By \ref{firstmoments} this holds for all time. Therefore we have by Lemma \ref{controlmarginals}
\[ d_{GTW,1}(\Pi_1[F_N], \Pi_1[H_N]) \leq d_{GTW,N}(F_N, H_N). \] Furthermore, by Theorem \ref{convergence}
\[ d_{GTW,N}(F_N(t), H_N(t)) \leq d_{GTW,N}(F_N(0), H_N(0))e^{-\frac{\mu}{2}t}. \] Now we use the chaoticity property and our control on tensorised functions form Lemma \ref{tensorisedfunctions} to get
\begin{align*}
d_{GTW,N}(F_N(0), H_N(0)) & \leq d_{GTW,N}(F_N(0), f^{\otimes N}) + d_{GTW, N}(f^{\otimes N},h^{\otimes N})\\&+ d_{GTW,N}(h^{\otimes N}, H_N(0)) \\
& = C_1 + d_{GTW,1}(f,h).
\end{align*} Here $C_1$ only depends on how close the initial data is to tensorised. Putting this together gives
\[ d_{GTW,1}(\Pi_1[F_N](t), \Pi_1[H_N](t)) \leq (d_{GTW,1}(f,h) + C_1) e^{-\frac{\mu}{2}t}. \] We do not have from our conditions that $C_1$ will decrease to 0 as $N \rightarrow \infty$, but since in this situation the real interest is just to choose any $f$-chaotic family we may as well have that $F_N(0) = f^{\otimes N}$ and similarly with $H$ which would dispense with the $C_1$ altogether.
\end{proof}

Now we would like to prove a theorem in the spirit of Theorem \ref{chaos} when we do not have $f$ and $h$ having zero mean initially. We cannot recover uniform estimates in $N$ but we can control the growth with $N$. We have from lemma \ref{controlmarginals} control of marginals by the function for the $\tilde{d}$ distance so we have
\[ \tilde{d}_k(\Pi_k[F_N], \Pi_k[H_N]) \leq \tilde{d} (F_N, H_N). \] Following this we would like to prove something in the spirit of lemma \ref{tensorisedfunctions} in order to control in the other direction.
\begin{lemma} Suppose we have $f$ and $h$ probability distributions on $\mathbb{R}$ with differentiable Fourier transforms.
If we define 
\[ n_f = \int |v|f(v) \mathrm{d}v, \] and let $M= \max \left\{ \frac{n_f}{|m_f|}, \frac{n_h}{|m_h|} \right\}$ then we have the following control by the first marginals for the $\tilde{d}$ distance on tensorised functions.
\[ \tilde{d}_N(f^{\otimes N}, h^{\otimes N}) \leq \tilde{d}_1(f,h) + M|m_f-m_h| \sqrt{N}. \]
\end{lemma}
\begin{proof}
Using the same bridging argument as before we see that
\begin{align*} &\hat{f}(\xi_1)\dots\hat{f}(\xi_N) - \hat{h}(\xi_1)\dots\hat{h}(\xi_N) - (m_f-m_h)\chi_N(\xi)\sum_k i \xi_k\\& = 
\sum_k \hat{f}(\xi_1)\dots\hat{f}(\xi_{k-1})(\hat{f}(\xi_k)-\hat{h}(\xi_k)-\chi_1(\xi_k)(m_f-m_h)i\xi_k)\hat{h}(\xi_{k+1})\dots\hat{h}(\xi_N) \\
&+ \sum_k \hat{f}(\xi_1)\dots\hat{f}(\xi_{k-1})(m_f-m_h)\chi_1(\xi_k)i\xi_k \hat{h}(\xi_{k+1})\dots\hat{h}(\xi_N) \\& - \chi_N(\xi)\sum_k (m_f-m_h)i\xi_k. \end{align*}
In order to complete the proof we want to bound the last term by something of the form
\[ M|m_f-m_h|\sqrt{N}|\xi|^2.\] Provided the radius of the set in which the $\chi$ are 1 is sufficiently large this will be true. So if we look at the last term where the $\chi$ are 1, we have
\[ (m_f-m_h)i \sum_k \xi_k \left(\hat{f}(\xi_1)\dots\hat{f}(\xi_{k-1})\hat{h}(\xi_{k+1})\dots\hat{h}(\xi_k)-1  \right). \] If instead we try and bound
\[ A=\frac{ \hat{f}(\xi_1)\dots \hat{f}(\xi_{k-1})\hat{h}(\xi_{k+1})\dots\hat{h}(\xi_N)-1 }{m_f \sum_{j<k} i \xi_j + m_h \sum_{k<j}i \xi_j} \leq M \] then we would have the bound
\begin{align*} &\left|\frac{\sum_k (\hat{f}(\xi_1)\dots\hat{f}(\xi_{k-1})\hat{h}(\xi_{k+1})\dots \hat{h}(\xi_N) -1)\xi_k (m_f-m_h)}{|\xi|^2}  \right| \\
& \leq M \frac{\left|\sum_{k=1}^N(m_f\sum_{j<k}i\xi_j + m_h \sum_{k<j}i\xi_j)\xi_k i (m_f-m_h) \right|}{|\xi|^2}  \leq M|m_f-m_h|\sqrt{N}.\end{align*} Therefore it remains to prove the bound on $A$, we do this first by noting that by Taylor expanding we can see that as $|\xi| \rightarrow 0, A \rightarrow 1$ and that as $|\xi| \rightarrow \infty, A \rightarrow 0$. A is differentiable everywhere except possibly $0$. Now we differentiate to get that at any stationary point of $A$ and for every $l<k$ we have 
\begin{align*} & \hat{f}(\xi_1)\dots\hat{f}'(\xi_l)\dots\hat{f}(\xi_{k-1})\hat{h}(\xi_{k+1})\dots\hat{h}(\xi_N) \left(m_f \sum_{j<k} i \xi_j + m_h \sum_{k<j} i \xi_j \right) \\ &= i m_f \left( \hat{f}(\xi_1) \dots \hat{f}(\xi_{k-1})\hat{h}(\xi_{k+1})\dots \hat{h}(\xi_N) -1  \right).  \end{align*} Substituting this into our expression for $A$ shows that at a stationary point
\[ A= \frac{1}{i m_f} \hat{f}(\xi_1)\dots\hat{f}'(\xi_l)\dots\hat{f}(\xi_{k-1})\hat{h}(\xi_{k+1})\dots \hat{h}(\xi_N) \leq M. \] This gives the claimed bound. It seems like there will be a problem if $m_f=0$ but if so we can always choose to differentiate in a direction so that we will get $m_h$ rather than $m_f$ and the cannot both be $0$. Here $C_1$, in the statement, only depends on the distance between the initial data and the tensorised functions, $C_2$ only depends on $g$ and $\chi$ and $C_3$ is a constant times $M|m_f-m_h|$ where $M$ is the maximum of $\int|v|f(v)\mathrm{d}v$ with the same quantity for $h$.\end{proof}
We can now prove the theorem
\begin{proof}[Proof of Theorem \ref{convwithcorrectionterm}]
This is found by putting together the convergence theorems and lemmas on distance control in exactly the same way as Theorem \ref{firstmoments}.
\end{proof}

If we move on to looking at the $T1$ distance we again have the bound on the $T1$ distance between marginals by the distance between the full function from Lemma \ref{controlmarginals}. We would like to be able to control the distance between tensorised functions by the marginals in order to give similar arguments to Theorem \ref{chaos} and Theorem \ref{convwithcorrectionterm}.

\begin{lemma}
\[ d_{T1,N}(f^{\otimes N}, h^{\otimes N}) \leq \sqrt{N} d_{T1,1}(f,h). \] Furthermore, the square root dependence is the best possible if $f,h$ have different means.
\end{lemma}
\begin{proof}
This follows a similar argument to the others
\begin{align*}
\sup_{\xi \neq 0}& \frac{|\hat{f}(\xi_1)\dots\hat{f}(\xi_N)-\hat{h}(\xi_1)\dots\hat{h}(\xi_N)|}{|\xi|}\\& \leq \sup_{\xi \neq 0} \frac{\sum_k |\hat{f}(\xi_1)\dots \hat{f}(\xi_{k-1}(\hat{f}(\xi_k)-\hat{h}(\xi_k))\hat{h}(\xi_{k+1})\dots\hat{h}(\xi_N)|}{|\xi|}\\
& \leq \sup_{\xi \neq 0} \sum_k \frac{|\hat{f}(\xi_k)-\hat{h}(\xi_k)|}{|\xi_k|} \frac{|\xi_k|}{|\xi|}\\
& \leq sup_{\xi \neq 0} \frac{|\hat{f}(\xi)-\hat{h}(\xi)|}{|\xi|}\sum_k \frac{|\xi_k|}{|\xi|} \\
& \leq \sqrt{N}\sup_{\xi \neq 0} \frac{|\hat{f}(\xi)-\hat{h}(\xi)|}{|\xi|}.
\end{align*}
The fact that the square root dependence is necessary for functions with different means can be seen by Taylor expanding 
\[ \frac{\hat{f}(\xi_1)\dots\hat{f}(\xi_N)-\hat{h}(\xi_1)\dots\hat{h}(\xi_N)}{|\xi|} \] around $\xi = 0$ then we can see that the limit as $\xi \rightarrow 0$ of this expression has modulus $\sqrt{N}|m_f-m_h|$.
\end{proof}
\begin{proof}[Proof of Theorem \ref{t1distfirstmarginal}]
Again we combine the convergence theorem that we have for the $T1$ distance with the control on distances as in Theorem \ref{firstmoments}.
\end{proof}

\section{Contraction in Wasserstein-2}
We can also show contraction of this model in Wasserstein distances using a simple coupling of two different systems. This coupling involves taking two of the coupled Kac's models and giving them simultaneous collisions with the same angle if it is an internal collision and the same angle and velocity of the external particle if it is an external collision. We can represent the stochastic process as an integral against several Poisson point processes. This is done in \cite{H16} and is helpful here to prove contraction for the energy process in Kac's model.
\begin{align}
V_{i,t} =& V_{i,0} + \lambda \sum_{j \neq i} \int_0^t \int_0^{2\pi} \left(V_{i,s^-}\cos \theta + V_{j, s^-}\sin \theta - V_{i, s^-}\right)\Pi_{i,j}(\mathrm{d}s, \mathrm{d}\theta) \\ &+ 2 \mu \int_0^t \int_{-\infty}^{\infty} \int_0^{2\pi}\left(V_{i,s^-} \cos \theta + w \sin \theta - V_{i,s^-}\right)\nu_i (\mathrm{d}s, \mathrm{d}w, \mathrm{d}\theta).
\end{align} Here $\Pi_{i,j}$ is a Poisson point process on $[0, \infty) \times [0, 2\pi]$ with intensity measure being $1/2\pi(N-1)$ times Lebesgue measure, and $\nu_i$ is a Poisson point process with intensity measure $g$ tensored with $1/2\pi(N-1)$ times Lebesgue. Using this representation we can prove contraction in Wasserstein-2.

\begin{proof}[Proof of Theorem \ref{wassersteincontraction}]
Using the representation above we can write out a similar formula for the difference between two solutions coupled by giving them the same driving Poisson processes. If we call this difference in the $i^{th}$ variable $\Delta_{i,t}$ then we can write
\begin{align*}
\Delta_{i,t}^2 =& \Delta_{i,0}^2 +\\&\lambda \sum_{j \neq i} \int_0^t \int_0^{2\pi} \left( \Delta_{i,s^-}^2 (\cos^2 \theta - 1) + \Delta_{j,s^-}^2 \sin^2 \theta +\right. \\& \left. 2\cos \theta \sin \theta \Delta_{i,s^-}\Delta_{j,s^-} \right)\Pi_{i,j}(\mathrm{d}s, \mathrm{d}\theta) \\ &+ 2 \mu \int_0^t \int_{-\infty}^{\infty} \int_0^{2\pi} \left( \Delta_{i,s^-}^2(cos^2 \theta - 1) + 2\Delta_{i,s^-} w \sin\theta \cos \theta \right) \nu(\mathrm{d}s, \mathrm{d}w, \mathrm{d}\theta).
\end{align*} Summing over $i$ and taking expectations gives
\begin{align*}
\frac{\mathrm{d}}{\mathrm{d}t} \mathbb{E}\left( \sum_{i=1}^n \Delta_{i,t}^2 \right) =& 2 \lambda (N-1) \frac{1}{2\pi} \int_0^{2\pi} (\cos^2 \theta + \sin^2 \theta -1) \mathrm{d}\theta  \mathbb{E}\left( \sum_{i=1}^n \Delta_{i,t}^2 \right) \\ &+ 2\mu \frac{1}{2\pi}\int_0^{2\pi} \int_{-\infty}^{\infty} g(w) (\cos^2 \theta - 1) \mathrm{d}\theta \mathrm{d}w  \mathbb{E}\left( \sum_{i=1}^n \Delta_{i,t}^2 \right)\\ =& -\mu  \mathbb{E}\left( \sum_{i=1}^n \Delta_{i,t}^2 \right).
\end{align*} Which gives the result after taking the infimum over possible couplings.
\end{proof}
We can also prove a similar controls over how Wasserstein distances behave in as the dimension goes to infinity. Here we write $\mathcal{W}_{p, d}$ to be the Wasserstein-2 distance related to the euclidean distance on $\mathbb{R}^d$. 
\begin{lemma}
If $\mu, \nu$ are measures on $\mathbb{R}$ with finite second moment then
\[ \mathcal{W}_{2,N}(\mu^{\otimes N}, \nu^{\otimes N}) = \sqrt{N}\mathcal{W}_{2,1}(\mu, \nu). \]
\end{lemma}
\begin{proof}
We know that there exists an optimal coupling, $\pi_1$ so that
\[ \mathcal{W}_{2,1}(\mu, \nu) = \left( \int_{\mathbb{R}^2} (x-y)^2 \pi_1(\mathrm{d}x, \mathrm{d}y)  \right)^{1/2} \] and an optimal coupling, $\pi_N$, such that
\[ \mathcal{W}_{2,N}( \mu^{\otimes N}, \nu^{\otimes N}) \left( \int_{\mathbb{R}^{2N}} \|\mathbf{x}-\mathbf{y}\|^2 \pi_N(\mathrm{d}\mathbf{x}, \mathrm{d}\mathbf{y}) \right)^{1/2}. \] Suppose that $\pi_N \neq \pi_1^{\otimes N}$ then we have that
\begin{align*} &\int \left( (x_1-y_1)^2 + \dots + (x_N-y_N)^2 \right)\pi_N(\mathrm{d}\mathbf{x}, \mathrm{d}\mathbf{y}) \\&< \int \left( (x_1-y_1)^2 + \dots + (x_N-y_N)^2 \right)\pi_1(\mathrm{d}x_1,\mathrm{d}y_1)\dots \pi_1(\mathrm{d}x_N, \mathrm{d}y_N) \\=& N \int (x-y)^2 \pi_1(\mathrm{d}x, \mathrm{d}y). \end{align*} Therefore, there exists some $k$ such that
\[ \int_{\mathbb{R}^{2N}}(x_k-y_k)^2 \pi_N(\mathrm{d}\mathbf{x}, \mathrm{d}\mathbf{y}) < \int_{\mathbb{R}^2} (x-y)^2 \pi_1(\mathrm{d}x, \mathrm{d}y). \] Since the integrand on the left hand side only depends on $x_k, y_k$ $\pi_N$ induces a coupling of $\mu$ and $\nu$ by projection onto the $k^{th}$ variables. The cost under this measure is strictly less that the optimal cost which is a contradiction. Hence, the optimal coupling is achieved by $\pi_1^{\otimes N}$. This gives that,
\begin{align*} \mathcal{W}_{2,N}&(\mu^{\otimes N}, \nu^{\otimes N}) \\&= \left(  \int \left( (x_1-y_1)^2 + \dots + (x_N-y_N)^2 \right)\pi_1(\mathrm{d}x_1,\mathrm{d}y_1)\dots \pi_1(\mathrm{d}x_N, \mathrm{d}y_N) \right)^{1/2} \\
&= \left( N \int (x-y)^2 \pi_1(\mathrm{d}x, \mathrm{d}y) \right)^{1/2} \\ &= \sqrt{N} \mathcal{W}_{2,1}(\mu, \nu).  \end{align*}
\end{proof}
\begin{lemma}
If $\mu_N$ and $\nu_N$ are symmetric probability distributions on $\mathbb{R}^N$ with finite second moment then
\[ \mathcal{W}_{2,1}(\Pi_1(\mu_N), \Pi_1(\nu_N)) \leq \frac{1}{\sqrt{N}} \mathcal{W}_{2,N}(\mu_N, \nu_N). \]
\end{lemma}
\begin{proof}
Suppose that $\pi_N$ is a coupling of $\mu_N$ and $\nu_N$ then the marginals of $\pi_N$ induce couplings of the marginals of $\mu_N$ and $\nu_N$.
\begin{align*}
&\left( \int  \left( (x_1-y_1)^2 + \dots + (x_N-y_N)^2 \right)\pi_N(\mathrm{d}\mathbf{x}, \mathrm{d}\mathbf{y})  \right)^{1/2} \\
=& \left( \int (x_1-y_1)^2 \pi_N(\mathrm{d}\mathbf{x}, \mathrm{d}\mathbf{y}) + \dots + \int (x_N-y_N)^2 \pi_N(\mathrm{d}\mathbf{x}, \mathrm{d}\mathbf{y}) \right)^{1/2} \\ \geq & \left( N \mathcal{W}_{2,1}(\Pi_1(\mu_N), \Pi_1(\nu_N))^2 \right)^{1/2} = \sqrt{N} \mathcal{W}_{2,1}(\Pi_1(\mu_N), \Pi(\nu_N)).
\end{align*}
\end{proof} Like with the earlier sections we can combine this behaviour with our contraction estimates to show uniform behaviour of the first marginal. For simplicity we only looked at tensorised initial data.

\begin{proof}[Proof of Theorem \ref{wassersteinmarginals}]
\begin{align*}
\mathcal{W}_{2,1}( \Pi_1(\mu_N(t)), \Pi_1(\nu_N(t))) &\leq \frac{1}{\sqrt{N}} \mathcal{W}_{2,N}(\mu_N(t), \nu_N(t))\\ &\leq \frac{1}{\sqrt{N}}e^{-\mu t/2} \mathcal{W}_{2,N}( \mu_0^{\otimes N}, \nu_0^{\otimes N}) \\ &= e^{- \mu t/2} \mathcal{W}_{2,1}(\mu_0, \nu_0).
\end{align*}
\end{proof}
\begin{remark}
These uniform estimates in $N$ combined with propagation of chaos means that the limit Boltzmann-Kac equation will also show exponential convergence to equilibrium in Wasserstein-2. This is very similar to the result shown in \cite{CLM15} in the Toscani distance.
\end{remark}
\begin{acknowledgements}
I would like to thank Cl\'{e}ment Mouhot for pointing me towards this problem and providing many useful discussions I also had many useful discussions with Amit Einav, Thomas Holding, Helge Dietert and Davide Piazzoli. Further, I would like to thank Amit Einav for reading the paper and many useful suggestions on the presentation and content.
\end{acknowledgements}

\bibliographystyle{spmpsci}      

\bibliography{kacbibliography.tex.bib}{}

\end{document}